\documentclass[11pt]{article}
\usepackage[margin=1in]{geometry}
\usepackage{amsmath,amsthm,amssymb}
\usepackage{graphicx,float,wrapfig}
\usepackage[dvipsnames]{xcolor}
\usepackage[ruled]{algorithm2e} 
\usepackage{hyperref}
\usepackage[nameinlink,capitalise]{cleveref}
\usepackage{tikz}
\usepackage{changepage}
\usepackage{mathtools}
\usepackage{enumerate}
\usepackage{thmtools,thm-restate}
\usepackage{libertine}
\usepackage{pgfplots}
\pgfplotsset{compat=1.18}

\makeatletter
\def\@fnsymbol#1{\ensuremath{\ifcase#1\or \dagger\or \ddagger\or
   \mathsection\or \mathparagraph\or \|\or **\or \dagger\dagger
   \or \ddagger\ddagger \else\@ctrerr\fi}}
\makeatother
\newcommand*\samethanks[1][\value{footnote}]{\footnotemark[#1]}

\newcommand{\pr}{\mathbf{Pr}} 
\newcommand{\one}{\mathbf{1}}%
\newcommand{\E}{\mathbb{E}}
\newcommand{\I}{\mathcal{I}}
\newcommand{\D}{\mathcal{D}}

\newcommand{\alg}{\mathsf{Alg}}
\newcommand{\reg}{\mathsf{regret}}
\newcommand{\rstar}{\widetilde{R}}
\newcommand{\ub}{R_q(\vec x)}
\renewcommand{\d}{\,\mathrm{d}}

\DeclareMathOperator*{\argmin}{arg\,min}

\Crefname{algocf}{Algorithm}{Algorithms}
\newtheorem{theorem}{Theorem}[section]
\newtheorem{observation}[theorem]{Observation}
\newtheorem{lemma}[theorem]{Lemma}

\newtheorem{proposition}[theorem]{Proposition}
\theoremstyle{definition}
\newtheorem{definition}[theorem]{Definition}
\newtheorem{example}[theorem]{Example}

\definecolor{myblue}{rgb}{0.15, 0.1, 0.95}
\definecolor{mygreen}{rgb}{0.15, 0.65, 0.25}
\definecolor{myred}{rgb}{0.75, 0.25, 0.15}

\definecolor{linkc}{rgb}{0.4, 0.3, 0.7}
\definecolor{citec}{rgb}{0.3, 0.6, 0.4}
\definecolor{urlc}{rgb}{0.3, 0.1, 0.2}
\hypersetup{
    linktocpage=true,
    breaklinks=true,
    colorlinks=true,
    linkcolor=linkc,
    citecolor=citec,
    urlcolor=urlc
}

\title{Additively Competitive Secretaries}

\author{Mohammad Mahdian\thanks{Google Research. Email: \texttt{\{mahdian,maojm\}@google.com}.} \hspace{-0.7em} \and Jieming Mao\samethanks[1] \hspace{-0.7em} \and Enze Sun\thanks{Department of Computer Science, the University of Hong Kong. Email: \texttt{sunenze@connect.hku.hk}.} \hspace{-0.7em} \and Kangning Wang\thanks{Department of Computer Science, Rutgers University. Email: \texttt{kn.w@rutgers.edu}.} \hspace{-0.7em} \and Yifan Wang\thanks{School of Computer Science, Georgia Institute of Technology. Email: \texttt{ywang3782@gatech.edu}.}}
\date{}

\begin{document}

\pagenumbering{gobble}
\thispagestyle{empty}
\maketitle
\begin{abstract}
In the secretary problem, a set of secretary candidates arrive in a uniformly random order and reveal their values one by one. A company, who can only hire one candidate and hopes to maximize the expected value of its hire, needs to make irrevocable online decisions about whether to hire the current candidate. The classical framework of evaluating a policy is to compute its worst-case competitive ratio against the optimal solution in hindsight, and there the best policy -- the ``$1/e$ law'' -- has a competitive ratio of $1/e$.

We propose an alternative evaluation framework through the lens of \emph{regret} -- the worst-case additive difference between the optimal hindsight solution and the expected performance of the policy, assuming that each value is normalized between $0$ and $1$. The $1/e$ law for the classical framework has a regret of $1 - 1/e \approx 0.632$; by contrast, we show that the class of ``pricing curves'' algorithms can guarantee a regret of at most $1/4 = 0.25$ (which is tight within the class), and the class of ``best-only pricing curves'' algorithms can guarantee a regret of at most $0.190$ (with a lower bound of $0.171$). In addition, we show that in general, no policy can give a regret guarantee better than $0.152$. Finally, we discuss other objectives in our regret-minimization framework, such as selecting the top-$k$ candidates for $k > 1$, or maximizing revenue during the selection process.
\end{abstract}

\medskip

\pagenumbering{arabic}

\section{Introduction}
In the domain of online decision making, people constantly face irrevocable decisions amid an uncertain future, where information about that uncertain future unfolds gradually. Central to this area is the theory of optimal stopping -- selecting one time to stop in order to maximize the associated reward. The \emph{prophet inequality} and the \emph{secretary problem} are two of the most prominent models of optimal stopping, each spawning numerous variants and fostering dedicated research in recent years.

In both the prophet inequality and the secretary problem, $n$ values are observed sequentially by a decision maker, who needs to select one value by making irrevocable decisions on whether to take the current value or not. In the prophet inequality, the values are drawn from independent prior distributions, which are known upfront to the decision maker. In the secretary problem -- a metaphor for a company trying to hire the best secretary -- the values are instead unknown and adversarial, but their arrival order is uniformly random. Traditionally, in both frameworks, the quality of a stopping rule (i.e., an algorithm or a policy) is measured through the \emph{competitive ratio} -- the worst-case ratio between the (expected) performance of the algorithm and the (expected) optimum in hindsight. The competitive ratio of the optimal algorithm in the prophet inequality is $1/2$, and that in the secretary problem is $1/e$.

In this work, we depart from the traditional competitive-ratio approach, and examine a different evaluation framework through the notion of \emph{regret}. Concretely, we assume that every value is within a certain known range, without loss of generality normalized to $[0, 1]$, and naturally define the regret as the worst-case \emph{additive difference} between the hindsight optimum and the performance of the algorithm. 

In fact, this framework of evaluating optimal stopping problems via regret has been proposed decades ago by Hill and Kertz, where they showed that the optimal regret is $1/4 = 0.25$ for the prophet inequality \cite{HK-Jour81}, and further examined various related settings \cite{HK-Jour82,HK-Jour83,Hil-Jour83}. However, to the best of our knowledge, there has not been any work in the literature that addresses the secretary problem in this additive-evaluation framework. Consequently, we explore the regret minimization objective for the secretary (random order) setting.

Why is this additive-evaluation framework valuable, especially in light of the extensive research on the competitive ratio of the secretary problem? We start by noting that our choice of studying this framework is influenced by its simplicity and elegance. We believe that our framework has the potential to inspire the development of novel algorithms, which, in turn, may offer practical utility or provide insights in real-world scenarios. Indeed, in many potential application domains, the (rough) knowledge about the range of possible values is available. The standard algorithm for the secretary problem -- the ``$1/e$ law'' -- has a regret of $1 - 1/e \approx 0.632$ in our framework,\footnote{The regret of the ``$1/e$ law'' is at most $1 - 1/e$ because its expected performance is at least $1/e$ times the optimum, and the optimum is at most $1$. Its regret is close to $1 - 1/e$ when there is a single value $1$ and many values close to $0$. {Furthermore, this hard instance applies to all ``$\alpha$-law'' algorithms with $\alpha \in [0, 1]$, indicating that $1 - 1/e$ is the best possible regret for the class of algorithms that first wait some fraction of time and, then, accept the first value above the maximum in the waiting phase.}} while our novel algorithms are significantly better, with regret guarantees as low as $0.190$. Additionally, in many situations, especially when the values correspond to money (which is inherently additive), minimizing the additive regret makes more sense than maximizing the competitive ratio. Regret as a benchmark is also robust when making multiple independent decisions, where the sum of regrets is still a meaningful measurement. We believe that our work serves as a starting point for a future line of research, on re-examining classical online decision models from this different perspective of regret minimization.


\subsection{Our Contribution}
\label{sec:our_contribution}

Our first contribution is to introduce the natural regret minimization framework for the secretary problem, formally in \cref{sec:prelim}. Then we begin our analysis by examining the regret of \emph{pricing curves}. A pricing curve is an algorithm that specifies a threshold which only depends on the arrival time of the current value, and the algorithm accepts a value if it is above the threshold. In \cref{sec:warmup}, we show that setting $f(t) = 1 - t$ to be the threshold function of a pricing curve gives a $1/4 = 0.25$ regret bound. This implies that the \emph{random-order arrival} in the secretary problem makes the regret minimization problem easier than adversarial-order arrival \cite{HK-Jour83} (discussed in \cref{sec:related}). Our proof hinges on the observation that the ``hardest'' instance for such an algorithm only contains two different values, which allows us to compute the regret directly. We also show a matching lower bound of $0.25$ for the class of pricing curves, meaning that we have to look beyond it to further improve the regret.

Next in \cref{sec:upper}, as our main result, we identify another class of simple algorithms that have good performance and simultaneously are    to analysis. We call them the \emph{best-only pricing curves}. Such an algorithm again specifies a threshold that only depends on the arrival time of the current value, and then it accepts a value if both (1) it is above the threshold, and (2) it is the highest value seen so far. We give an upper bound of $0.190$ and a lower bound of $0.171$ for the optimal best-only pricing curve. The key for these bounds is a series of observations that lead to relaxed expressions of regret with only a small number of variables. 

\definecolor{color1}{rgb}{0.3, 0.3, 0.8}
\definecolor{color2}{rgb}{0.8, 0.3, 0.3}
\definecolor{color3}{rgb}{0.9, 0.6, 0.1}

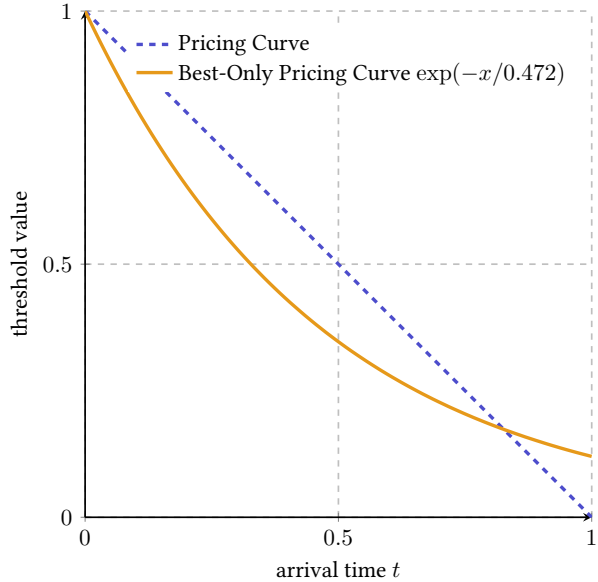
\begin{figure}[tbh]
\label{fig:curve}
\centering
\begin{tikzpicture}[scale = 0.8,
declare function={
func(\x) = (0<=\x) * (\x<1) * exp(\x/-0.5562);
}]
\begin{axis}[
    axis lines=left,
    xlabel={arrival time $t$},
    ylabel={threshold value},
    xmin=0, xmax=1,
    ymin=0, ymax=1,
    xtick={0, 0.5, 1},
    ytick={0, 0.5, 1},
    legend pos=north east,
    legend cell align={left},
    legend style={draw=none},
    ymajorgrids=true,
    xmajorgrids=true,
    grid style={dashed, thick},
    axis line style ={thick},
    width=10cm,
    height=10cm,
    samples=100
]

\addplot[color1, ultra thick, dashed, domain=0:1] {1-x};
\addlegendentry{Pricing Curve}


\addplot[color3, ultra thick, domain=0:1] {exp(\x/ -0.472)};
\addlegendentry{Best-Only Pricing Curve $\exp(-x/0.472)$}
\end{axis}
\end{tikzpicture}
\caption{Optimal pricing curves for both algorithms}
\end{figure}

Let us look further into a best-only pricing curve, in order to intuitively understand why it can provide improvement over pricing curves. In \Cref{fig:curve}, we illustrate the threshold function of the optimal pricing curve (blue dashed) and that of the best-only pricing curve with regret at most $0.190$ (yellow solid). Firstly, note that the latter has a lower threshold function, which enables it to pick a sufficiently large value in cases where all values are small. Secondly, the best-only pricing curve uses the previously-seen maximum value as a second threshold, allowing the algorithm to ``postpone'' the time of choosing a frequently appearing value and hence giving the algorithm more chances to accept a higher value.

As a concrete example, consider an instance with a single value of $1$ and near-infinitely many values of $0.5$ with small noises to break ties. The pricing curve algorithm will accept the value of $1$ if it arrives before $t = 0.5$, and otherwise will take a value of $0.5$. This leads to its relatively large regret of $0.25$. For the best-only pricing curve, although the threshold function drops below $0.5$ at time $t \approx 0.3$, the early-arriving values of $0.5$ with more positive noises act as a second threshold and prevent the algorithm from choosing a value of $0.5$ with a less positive noise. As a result, with more than $66\%$ chance, the algorithm skips all values of $0.5$ and chooses the value of $1$, improving the regret to less than $0.17$.

To complement our positive results, in \cref{sec:lower}, we provide a general lower bound of $0.152$ that holds for any algorithm. The main idea of proving the lower bound is to show that the optimal algorithm has a relatively simple structure, albeit possibly with a large number of parameters. Then, we construct a hard instance with a manageably small size, compute the optimal (exponential time) algorithm for the hard instance, and show that the regret is at least $0.152$. 

Finally, in \cref{sec:revenue}, we discuss the multiple-choice secretary problem and also the related objective of \emph{revenue maximization} in the regret minimization framework. In the multiple-choice secretary problem, a company can hire $k$ secretaries and obtain the sum of their values. We show that the bound of the optimal regret in this setting is $\Theta(\sqrt{k})$ as $k \to \infty$. In the revenue maximization problem, a seller is selling an item to $n$ sequential buyers, where buyer $i$ has valuation $x_i \in [0, 1]$ for the item. The seller needs to post a price of $p_i$ before buyer $i$ arrives. The first buyer with $x_i \geq p_i$ takes the item, and contributes $p_i$ to the revenue. The goal of the seller is again to minimize the regret -- the difference between the highest possible revenue, which is the maximum of all $x_i$'s, and the revenue achieved by the algorithm. For this objective, we show that the optimal regret is tight at $1/e$ for many variants of the model: the buyers' valuations can be either adversarial, stochastic, or i.i.d.\@ stochastic, and the arrival order can be adversarial or uniformly random.

\subsection{Further Related Works}
\label{sec:related}

\paragraph{Regret Minimization for Optimal Stopping Problems.}
The objective of minimizing regret in optimal stopping has a long history, beginning with the series of works by Hill and Kertz \cite{HK-Jour81,HK-Jour82,HK-Jour83,Hil-Jour83}, who analyzed the standard, i.i.d., correlated, and order-selection versions of the prophet inequality.
In the correlated case, \cite{HK-Jour83} proved that when values are normalized to $[0,1]$, the optimal regret of any stopping rule is $1/e$.
For independent distributions, \cite{HK-Jour81} established a tight regret bound of $1/4$, while in the i.i.d.\ setting, the optimal bound improves to approximately $0.111$ \cite{HK-Jour82}.
These lower bounds naturally extend to our framework as special cases, although our analysis yields a stronger lower bound of $0.152$, demonstrating a strictly harder setting.

Another line of research focuses on the \emph{multisecretary problem}, initiated by Arlotto and Gurvich \cite{ArlottoGurvich2019}. 
In this model, a sequence of $T$ i.i.d.\ values drawn from a fixed distribution (independent of $T$) arrives online, and the decision maker may select up to $k$ of them, often with $k$ proportional to $T$. 
The objective is to minimize the \emph{regret}, defined as the difference between the expected offline optimum and the algorithm’s expected reward, in the limit as $T \to \infty$. 
There are a number of follow-up papers exploring this direction, including \cite{BesbesKanoriaKumar2022EC,Bray2024,VeraBanerjeeGurvich2021,BumpensantiWang2020}, which study extensions to multi-type settings, continuous distributions, and other online allocation frameworks in operations research.

\paragraph{The Secretary Problem and Other Random-Order Models.}
The classical secretary problem dates back to early works such as \cite{Gar60,lindley1961dynamic,Dyn-63,ferguson1989solved}.
A wide range of variants have been explored, including minimizing relative rank \cite{chow1964optimal}, robust mixtures of random and adversarial orders \cite{DBLP:conf/innovations/BradacG0Z20}, and models with advice or samples \cite{DBLP:conf/soda/KaplanNR20,DBLP:conf/soda/CorreaCFOT21}.
The model has also been extended to combinatorial domains.
Examples include the matroid secretary problem \cite{DBLP:conf/soda/BabaioffIK07,DBLP:journals/jacm/BabaioffIKK18,DBLP:conf/soda/ChakrabortyL12,DBLP:conf/focs/Lachish14,DBLP:journals/mor/FeldmanSZ18}, and its constant-competitive versions for special matroid classes \cite{DBLP:journals/algorithmica/DimitrovP12,DBLP:conf/soda/ImW11}.
Other generalizations consider bipartite and general matchings \cite{DBLP:conf/icalp/KorulaP09,DBLP:conf/esa/KesselheimRTV13,DBLP:conf/sigecom/EzraFGT22}, downward-closed systems \cite{DBLP:conf/stoc/Rubinstein16}, and submodular objectives \cite{DBLP:journals/talg/BateniHZ13}.
More recently, the \emph{prophet secretary} model—combining prophet inequalities with random arrival—has attracted significant attention \cite{DBLP:journals/siamdm/EsfandiariHLM17}.
Improved competitive ratios have been obtained in successive works \cite{DBLP:conf/sigecom/AzarCK18,DBLP:journals/mp/CorreaSZ21,harb2023fishing}, while hardness bounds are given in \cite{giambartolomei2023prophet,DBLP:journals/mp/CorreaSZ21,DBLP:conf/sigecom/BubnaC23}.
Extensions to matroids and combinatorial auctions have also been studied \cite{DBLP:conf/soda/EhsaniHKS18}.

\paragraph{Other Optimal Stopping Problems.}
The \emph{prophet inequality}, originating from classical works \cite{Krengel-Journal77,Krengel-Journal78,Samuel-Annals84,HK-Jour81}—is a cornerstone in optimal stopping and mechanism design, with numerous extensions to matroids, matchings, and other combinatorial domains \cite{DBLP:conf/stoc/ChawlaHMS10,DBLP:journals/geb/KleinbergW19,DBLP:conf/soda/FeldmanGL15,DBLP:conf/soda/RubinsteinS17,DBLP:conf/sigecom/EzraFGT20}.
A closely related line of research studies \emph{revenue maximization} through sequential posted pricing (SPP), where each arriving buyer faces a take-it-or-leave-it price.
Variants of SPP achieve constant-factor approximations to Myerson’s optimal mechanism for single-item sales \cite{HartlineBook,DBLP:conf/stoc/ChawlaHMS10,DBLP:conf/soda/Yan11,DBLP:journals/siamcomp/Alaei14} and extend to multi-item settings with unit-demand buyers \cite{DBLP:conf/stoc/ChawlaHMS10,DBLP:journals/geb/KleinbergW19,DBLP:journals/geb/ChawlaMS15}.


\section{Preliminaries}
\label{sec:prelim}

Throughout our work, we focus on the following definition of the secretary problem. An adversary specifies $n$ values $x_1, x_2, \ldots, x_n$ with $1 \geq x_1 \geq x_2 \geq \dots \geq x_n \geq 0$. Each value $x_i$ arrives at time $t_i$ drawn from $U[0, 1]$, the uniform distribution supported on the interval $[0, 1]$. Note that $t_i$'s are almost surely distinct. An online algorithm observes each value in the increasing order of their arrival time, and needs to immediately decide whether to pick the observed value or to discard it. An algorithm can only pick one value, and for consistency, we assume that an algorithm always picks a value of $0$ at the end if it has not picked any value yet. (Alternatively, we can always modify an algorithm so that it never skips the last value.)

Let $\alg(\vec{x})$ be the expectation of the picked value of the algorithm $\alg$ when the values are $\vec{x}$. Our objective is to minimize the worst-case \emph{regret}, which is defined as
\[
\reg \coloneqq \sup_{\vec{x}} \reg(\vec{x}), \quad \text{where} \quad \reg(\vec{x}) \coloneqq x_1 - \E[\alg(\vec{x})].
\]

We note that values arriving at random time drawn from $U[0, 1]$ is equivalent to values arriving in a uniformly random order, since the time vector beyond the order does not provide any useful information. An algorithm that uses this information in its description can simply freshly draw the vector of arrival time from $(U[0, 1])^n$. We note that the step to simplify the analysis has been done in related work, such as that of \cite{DBLP:journals/mp/CorreaSZ21,DBLP:conf/soda/CorreaCFOT21}.

\section{Warm-up: The Regret of Pricing Curves}
\label{sec:warmup}
In this section, we discuss \emph{pricing curves} -- a class of simple algorithms that, as we will show, can achieve a regret bound of $0.25$. A pricing curve specifies a (measurable) threshold function $f(t) \colon [0, 1] \to [0, 1]$ that defines a threshold value at each possible arrival time. It then picks the first-arrived value that is above the threshold at its arrival time.

\begin{algorithm}[tbh]
\caption{Pricing Curve}
\label{alg:pc}
Let $f(t) \colon [0, 1] \to [0, 1]$ be its threshold function \\
\For{$\tau = 0 \to 1$ where $x_i$ arrives at $t_i = \tau$}
{
    \uIf{$x_i \geq f(t_i)$}{Pick $x_i$ and \textbf{exit}}
    \Else{Skip $x_i$}
}
\end{algorithm}

\begin{theorem}
\label{thm:mainpc}
    \Cref{alg:pc} with $f(t) = 1 - t$ achieves a regret of ~$0.25$.
\end{theorem}

\begin{proof}
We consider a more difficult variant of our setting, and show that the regret of the pricing curve with $f(t) = 1 - t$ is at most $0.25$ in that variant, hence implying our theorem statement.

In the variant, $x_1$ arrives at $t_1 \sim U[0, 1]$, but $t_2, t_3, \ldots, t_n$ are determined upfront by an adversary. We now characterize the hardest instance for our algorithm in this variant.

First, notice that if $x_i < f(t_i)$ for some $i \geq 2$, we can simply remove this $x_i$ without changing the behavior or the regret of our algorithm. We do this for all such $x_i$'s.

After that, let $j \coloneqq \argmin_{i \geq 2} t_i$. (If all $x_i$'s with $i \geq 2$ have been removed and hence the index $j$ does not exist, we can add back $x_2 = 0$ with $t_2 = 1$ without changing the analysis.) Again, we can remove all values except $x_1$ and $x_j$ without changing the behavior or the regret of our algorithm.

Now we analyze the regret after these simplifications. Observe that our algorithm picks $x_1$ if and only if $t_1 \in [1 - x_1, t_j]$, and picks $x_j$ otherwise. Therefore, the regret of our algorithm is
\[
\reg ~=~ (x_1 - x_j) \cdot (1 - (t_j - (1 - x_1))) ~\leq~ (x_1 - x_j) \cdot (1 - x_1 + x_j) ~\leq~ 0.25,
\]
where in the first inequality, we use the fact that $x_j \geq 1 - t_j$.
\end{proof}

\Cref{thm:mainpc} shows that a pricing curve with a simple threshold function achieves a lower regret than that of $1 - 1/e$ of the ``$1/e$ law'', and the optimal regret of $1/e$ in the adversarial-order setting. As it turns out (in \cref{thm:pc_lower}), this regret bound of $0.25$ is optimal among all pricing curves.

\begin{theorem}
\label{thm:pc_lower}
    \Cref{alg:pc} has a regret of at least~$0.25$ for any pricing curve $f(t)$.
\end{theorem}

\begin{proof}
    Consider the following two instances:
    \begin{itemize}
        \item Instance (1): $n = 1$ and $x_1 = 0.5$.
        \item Instance family (2): $x_1 = 1$ and $x_2 = \dots = x_n = 0.5$.
    \end{itemize}
    We will show that \Cref{alg:pc} with any fixed threshold function $f(t)$ has a regret of at least $0.25$ for either Instance (1) or Instance family (2).

    Let $\alpha \coloneqq \mu(\{t \in [0, 1] : f(t) < 0.5\})$ denote the fraction (Lebesgue measure) of time where a value of $0.5$ is above the threshold. If $\alpha \leq 0.5$, then the pricing curve incurs regret of at least $(1 - \alpha)x_1 \geq 0.25$ in Instance (1). Otherwise (if $\alpha > 0.5$), note that there is a probability of $0.5$ that $x_1 = 1$ arrives at $t_1 \geq 0.5$ in Instance family (2). However, since $\mu(\{t \in [0, 0.5] : f(t) \geq 0.5\}) \geq \alpha - 0.5 > 0$, we know that as $n \to \infty$, the probability that the pricing curve has accepted some value by time $t = 0.5$ approaches $1$. Therefore, the regret of the pricing curve in Instance family (2) is at least $0.5(x_1 - x_2) = 0.25$.
\end{proof}

\section{Main Result: The Regret of Best-Only Pricing Curves}
\label{sec:upper}
By now, we have seen that the pricing curves can achieve a regret of $0.25$ but not better. That leads to a natural question: Can we achieve a lower regret using other -- hopefully still simple -- algorithms?

In this section, we move our attention to the class of \emph{best-only pricing curves}, and show that the regret of one such algorithm is at most $0.190$. We complement our result by a lower bound of $0.171$ for any algorithm in this class.

Similar to a pricing curve, a best-only pricing curve specifies a threshold function $f(t) : [0, 1] \rightarrow [0, 1]$. When the value $x_i$ arrives at time $t_i$, we accept it if and only if both
\begin{itemize}
\item $x_i \geq f(t_i)$, and
\item $x_i$ is the largest value that has arrived so far.\footnote{To avoid tie-breaking, we add an arbitrarily small noise to each value $x_i$, which makes all values almost surely unique, without affecting the performance guarantee of the algorithm.}
\end{itemize}
The following \Cref{alg:bopc} formally defines the algorithm we run. Note that it is possible that \Cref{alg:bopc} does not accept any value. In this case, we assume that the value we choose is $0$ and the regret is $x_1$.

\begin{algorithm}[tbh]
\caption{Best-Only Pricing Curve}
\label{alg:bopc}
Let $f(t) \colon [0, 1] \to [0, 1]$ be its threshold function \\
\For{$\tau = 0 \to 1$ where $x_i$ arrives at $t_i = \tau$}
{
    \uIf{$x_i \geq f(t_i)$ and $x_i$ is the largest value that has arrived so far}{Pick $x_i$ and \textbf{exit}}
    \Else{Skip $x_i$}
}
\end{algorithm}

\subsection{Regret Relaxation}
\label{sec:reg-relax}

Fix an instance with values $x_1, \ldots, x_n$, where $1 \geq x_1 \geq x_2 \geq \cdots \geq x_n \geq 0$. We first provide some preliminary analysis for \Cref{alg:bopc}, which will indicate that the hard instance for \Cref{alg:bopc} after relaxing the regret function has a relatively simple structure.

For simplicity of our analysis, in this section, we only consider \Cref{alg:bopc} with a pricing curve $f(t)$ with the following properties:
\begin{itemize}
    \item $f(t)$ is continuous.
    \item $f(t)$ is monotonically strictly decreasing.
    \item $f(0) = 1$.
\end{itemize}
With the above assumptions, the following observation suggests that it is sufficient only to consider the instances with values $x_1 \geq x_2 \geq \cdots \geq x_n \geq f(1)$.
\begin{observation}
    \label{obs:geq}
    Fix the pricing curve $f(t)$ that satisfies continuity, monotonicity, and that $f(0) = 1$. Let $x_1 \geq x_2 \geq \cdots \geq x_n$ be an input instance that maximizes the expected regret of \Cref{alg:bopc} with pricing curve $f(t)$. Then, there exists a setting of $\vec x = (x_1, \ldots, x_n)$ that satisfies:
    \begin{itemize}
        \item $n \geq K$ for some integer $K \geq 10$
        \item $x_1 \geq x_2 \geq \cdots \geq x_n \geq f(1)$.
    \end{itemize}
\end{observation}
\begin{proof}
It is sufficient to show that for any instance that does not satisfy the properties in \Cref{obs:geq}, there exists another instance with a regret no smaller than the current instance. Fix $\vec x = (x_1, \cdots, x_n)$ to be an instance that does not satisfy the above properties. Consider going through the following three inspections sequentially.
\begin{itemize}
    \item \textit{Inspection 1: $x_1 < f(1)$.} Note that in this case, the regret of the instance is $x_1 < f(1)$. In this case, we update $x_1$ to $f(1)$, and the regret becomes $f(1)$. 
    \item \textit{Inspection 2: $x_n < f(1)$.} In this case, we remove $x_n$ from the instance, and the regret of the instance remains unchanged. We ask the instance to go through Inspection 2 for multiple times.
    \item \textit{Inspection 3: $n \geq 1$ but $n < K$.} In this case, we add multiple $f(1)$ into the instance until $n \geq K$. Note that we require $f(t)$ to be strictly decreasing. Therefore, a value $f(1)$ can only be chosen with probability $0$, and adding a constant number of $f(1)$ into the instance does not change the regret.
\end{itemize}
Note that after passing the above three inspections, the regret of the instance does not decrease, but the properties in \Cref{obs:geq} are satisfied, which finishes the proof.
\end{proof}

With \Cref{obs:geq}, in the following of the section we only consider input instance $x_1 \geq x_2 \geq \cdots \geq x_n \geq f(1)$ with $n \geq 10$.

Let $g(x) \colon [f(1), 1] \rightarrow [0, 1]$ be the ``inverse function'' of $f(t)$, where, formally, $g(x) \coloneqq \inf \{t \in [0, 1]: x \geq f(t)\}$. In other words, $g(x_i)$ is the earliest time that the algorithm can accept $x_i$. (Note that $t_i \neq g(x_i)$ almost surely, and hence the corner case where $t_i = g(x_i)$ can be ignored.)
For simplicity, we use $\theta_i$ to denote $g(x_i)$, the earliest time that the algorithm can accept $x_i$. We have $\theta_n = 1$.

Next, we give a regret upper-bound for the instance $\vec x$. Perhaps surprisingly, we can simplify the regret for a fixed instance into the more manageable expression in \cref{lem:bopc_eq}.

\begin{lemma}
\label{lem:bopc_eq}
It holds that
\[
\reg(\vec x) ~=~ x_1 \cdot \theta_1 + \left(\sum_{i = 2}^{n - 1} (x_1 - x_i) \cdot \frac{1}{i} \cdot (1-\theta_i)^i\right) - \left(\sum_{i = 2}^{n - 1} (x_1 - x_i) \cdot \sum_{k = i+1}^{n-1} \left(\frac{1}{k-1} - \frac{1}{k}\right) \cdot (1-\theta_k)^k\right).
\]
\end{lemma}
\begin{proof}
Let $A_0$ be the event that the algorithm does not accept any value. This event $A_0$ occurs if and only if $x_1$ arrives before $\theta_1$, otherwise we must accept some value before or at the time that $x_1$ arrives. Therefore, $\Pr[A_0] = \theta_1$.

Let $A_i$ be the event that the algorithm accepts $x_i$. Now we fix $i$ and calculate the probability that $A_i$ occurs. Recall that $t_i \sim U[0, 1]$ is the arrival time of $x_i$. If $t_i < \theta_i$, then $x_i$ will be rejected. Otherwise, we have $t_i \in [\theta_k, \theta_{k+1})$ for some $k \geq i$. When $A_i$ occurs, there are two cases:
\begin{enumerate}
\item For all $s \in [k] \setminus \{i\}$, the value $x_s$ arrives after $t_i$. This happens with probability $(1-t_i)^{k-1}$.
\item Otherwise, let $j := \arg \min \{s \in [k] \setminus \{i\}: x_s \text{ arrives before } t_i\}$. In this case, $A_i$ occurs if and only if $j > i$ and $x_j$ arrives before $\theta_j$. This probability is $(1 - t_i)^{j - 2} \cdot \theta_j$.
\end{enumerate}
Therefore, through algebraic manipulation, we have
\begin{align*}
\Pr[A_i] &~=~ \sum_{k = i}^{n-1} \int_{\theta_k}^{\theta_{k+1}} \left((1-\tau)^{k-1} + \sum_{j = i+1}^k \theta_j \cdot (1-\tau)^{j-2}\right) \d \tau \\
&~=~ \left(\sum_{k=i}^{n-1} \int_{\theta_k}^{\theta_{k+1}} (1-\tau)^{k-1} \d \tau\right) + \left(\sum_{j = i+1}^{n-1} \theta_j \cdot \int_{\theta_j}^1 (1-\tau)^{j-2} \d \tau\right) \\
&~=~ \left(\sum_{k=i}^{n-1} \frac{1}{k} \left((1-\theta_k)^k - (1-\theta_{k+1})^k\right)\right) + \left(\sum_{j = i+1}^{n-1} \frac{\theta_j}{j-1} \cdot (1-\theta_j)^{j-1}\right) \\
&~=~ \sum_{k = i}^{n-1} \frac{1}{k} \left((1-\theta_k)^k - (1-\theta_{k+1})^k + \theta_{k+1}(1-\theta_{k+1})^k\right) \\
&~=~ \sum_{k = i}^{n-1} \frac{1}{k} \left((1-\theta_k)^k -(1-\theta_{k+1})^{k+1}\right) \\
&~=~ \frac{1}{i} \cdot (1-\theta_i)^i - \sum_{k = i+1}^{n-1} \left(\frac{1}{k-1} - \frac{1}{k}\right) \cdot (1-\theta_k)^k.
\end{align*}

Now we calculate the regret of our algorithm by plugging in these formulas.
\begin{align*}
\reg(\vec x) ~=~& (x_1 - 0) \cdot \Pr[A_0] + \sum_{i = 2}^{n} (x_1 - x_i) \cdot \Pr[A_i] \\
~=~& x_1 \cdot \theta_1 + \sum_{i = 2}^{n} (x_1 - x_i) \cdot \left(\frac{1}{i} \cdot (1-\theta_i)^i - \sum_{k = i+1}^{n-1} \left(\frac{1}{k-1} - \frac{1}{k}\right) \cdot (1-\theta_k)^k\right) \\
~=~& x_1 \cdot \theta_1 + \left(\sum_{i = 2}^{n} (x_1 - x_i) \cdot \frac{1}{i} \cdot (1-\theta_i)^i\right) \\ 
& - \left(\sum_{i = 2}^{n} (x_1 - x_i) \cdot \sum_{k = i+1}^{n} \left(\frac{1}{k-1}  - \frac{1}{k}\right) \cdot (1-\theta_k)^k\right),
\end{align*}
which matches the lemma statement.
\end{proof}

Next, we further relax the expression provided by \Cref{lem:bopc_eq} to a relatively simple form. We define the following.
\begin{definition}
We use $R_q(\vec x)$ (where $q = 2, 3, \ldots$) to present a relaxation of $\reg$:
\begin{align*}
R_q(\vec x) ~:=~&  x_1 \cdot \theta_1 + \left(\sum_{i = 2}^{n} (x_1 - x_i) \cdot \frac{1}{i} \cdot (1-\theta_i)^i\right) \\ 
& - \left(\sum_{i = 2}^{q - 1} (x_1 - x_i) \cdot \sum_{k = i+1}^{n} \left(\frac{1}{k-1}  - \frac{1}{k}\right) \cdot (1-\theta_k)^k\right)\\
~\geq~& \reg(\vec x),
\end{align*}
where the inequality holds by observing that we define $R_q(\vec x)$ by dropping some negative terms in.
\end{definition}

With the definition of $R_q(\vec x)$, we give the following \Cref{lem:maximized_at_eq_new}, which suggests that the hardest instance $\vec x$ that maximizes the relaxed regret function $R_q(\vec x)$ has a simple structure.
\begin{lemma}
\label{lem:maximized_at_eq_new}
For fixed $n, x_1, x_2, \dots, x_{q-1}$, if the pricing curve satisfies that $f(x)$ is non-increasing and convex, the function $R_q$ is maximized when $x_q = x_{q+1} = \cdots = x_{n}$.
\end{lemma}

\begin{proof}
We calculate $\frac{\partial R_q(\vec x)}{\partial x_k}$ for $j \leq k \leq n$. Recall that
\begin{align*}
R_q(\vec x) ~=~ &x_1 \cdot \theta_1 + \left(\sum_{i = 2}^{n} (x_1 - x_i) \cdot \frac{1}{i} \cdot (1-\theta_i)^i\right)  \\ & \quad - \left(\sum_{i = 2}^{q - 1} (x_1 - x_i) \cdot \sum_{k = i+1}^{n} \left(\frac{1}{k-1} - \frac{1}{k}\right) \cdot (1-\theta_k)^k\right).
\end{align*}
We have
\begin{align*}
\frac{\partial R_q(\vec x)}{\partial x_k} ~=~ & \frac{\partial}{\partial x_k}\left[\left((x_1 - x_k) \cdot \frac{1}{k} - \sum_{i = 2}^{q - 1} (x_1 - x_i) \cdot \frac{1}{k(k - 1)}\right) \cdot (1-\theta_k)^k\right]\\
~=~ & -\left((x_1 - x_k) - \sum_{i = 2}^{q - 1} (x_1 - x_i) \cdot \frac{1}{k - 1} \right) \cdot (1-\theta_k)^{k - 1} \cdot \theta'_k - \frac{1}{k} \cdot (1-\theta_k)^k,
\end{align*}
where $\theta'_k$ denotes $g'(x_k)$. We note that $\frac{\partial R_q(\vec x)}{\partial x_k}$ does not depend on $x_{k'}$ for $k' \in [q, n] \setminus \{k\}$.

For $k \in [q, n]$ and $f(1) \leq x \leq x_{q-1}$, define
\[
u_k(x) := - k \cdot \left((x_1 - x) -  \sum_{i = 2}^{q - 1} (x_1 - x_i) \cdot \frac{1}{k - 1} \right) \cdot g'(x) - (1-g(x)).
\]
Recall that we defined $\theta_k = g(x_k)$ and $\theta'_k =g'(x_k)$. Then, we have
\begin{align*}
    &u_k(x_k) := - k \cdot \left((x_1 - x_k) - \sum_{i = 2}^{q - 1} (x_1 - x_i)\cdot \frac{1}{k - 1} \right)  \cdot \theta'_k - (1-\theta_k)\\
    &\text{and} \qquad \frac{\partial R_q(\vec x)}{\partial x_k} = \frac{1}{k} \cdot u_k(x_k) \cdot (1 - \theta_k)^{k - 1}.
\end{align*}

We next examine the condition of $\frac{\partial \ub}{\partial x_k} = 0$, which is equivalent to $u_k(x_k) = 0$. Note that have
\[
u'_k(x) =  (k + 1) \cdot g'(x)- k \cdot \left((x_1 - x) + \sum_{i = 2}^{q - 1} (x_1 - x_i) \cdot \frac{1}{k - 1} \right) \cdot g''(x) < 0.
\]
The inequality holds because we have $g'(x) < 0$ and $g''(x) \geq 0$, which is true because we assumed $f(t)$ is strictly decreasing and convex, and $g(x)$ is the inverse function of $f(t)$, and  
\[
(x_1 - x) - \sum_{i = 2}^{q - 1} (x_1 - x_i) \cdot \frac{1}{k - 1} \geq (x_1 - x_{q-1}) - \frac{q-2}{k-1} (x_1 - x_{q-1}) \geq 0.
\]

Therefore, function $u_k(x)$ is decreasing within $[f(1), x_{q-1}]$. Additionally, notice that
\begin{align*}
    &u_k(f(1)) = -k \cdot \left((x_1 - f(1)) - \sum_{i = 2}^{q - 1} (x_1 - x_i)\cdot \frac{1}{k - 1} \right) \cdot g'(f(1)) ~\geq~ 0.
\end{align*}
Therefore, either equation $u_k(x) = 0$ has a unique solution within $[f(1), x_{q-1}]$, or $u_k(x) > 0$ for every $x \in [f(1), x_{q-1}]$. Let $z_k$ be the unique solution of $u_k(x) = 0$ within $[f(1), x_{q-1}]$. If the solution does not exist, we define $z_k = x_{q-1}$. Then, $\frac{\partial \ub}{\partial x_k}(x)$ is non-negative for $x \in [f(1), z_k]$, and is non-positive for $x \in (z_k, x_{q-1}]$.

We further observe that
\begin{align*}
    u_{k+1}(z_k) - u_k(z_k) &~=~ -(x_1 - z_k) \cdot g'(z_k) + \left(\frac{k+1}{k} - \frac{k}{k-1}\right) \cdot \sum_{i = 2}^{q - 1} (x_1 - x_i) \cdot g'(z_k)  \geq 0,
\end{align*}
i.e., $u_{k+1}(z_k) \geq u_k(z_k)$. If $u_k(z_k) = 0$, then we have $u_{k+1} (z_k) \geq 0$, and there must be $z_k \leq z_{k+1}$, because  function $u_{k+1}(x)$ is decreasing. Otherwise, we have $z_k = x_{q-1}$, and $u_{k+1}(x_{q-1}) > 0$, so $z_{k+1} = x_{q-1}$ should hold. Therefore, in both cases, we have $z_k \leq z_{k+1}$.

Finally, we finish our proof by showing the following claim: given vector $\vec x = (x_1, \ldots, x_n \geq f(1))$, there exists a vector $\vec w  = (w_1, \ldots, w_n)$ that satisfies the following conditions:
\begin{itemize}
    \item For $k \in [q-1]$, $w_k = x_k$.
    \item $w_{q} = w_{q+1} = \dots =w_{n}$.
    \item $R_q(\vec x) \leq R_q(\vec w)$.
\end{itemize}
We construct $\vec w$ by discussing the following three cases. \cref{fig:var} illustrates the relative rankings among different variables.

\begin{figure}[H]
\centering
\begin{tikzpicture}[yscale=1.4]

\node at (-1, 2) {Case 1:};
\draw[thick, ->] (0, 2) -- (10, 2);
\foreach \x in {1.8, 3.8, 5.3, 7.3} 
    \fill[blue] (\x, 2) circle (2pt);

\node[below] at (1.8, 1.9) {\textcolor{red}{$z_k$}};
\node[below] at (3.8, 1.9) {$z_{n}$};
\node[below] at (5.3, 1.9) {$x_{n}$};
\node[below] at (7.3, 1.9) {\textcolor{red}{$x_k$}};

\draw[->, thick, blue] (3.8, 2.35) -- (3.8, 2.1);

\node at (-1, 1) {Case 2:};
\draw[thick, ->] (0, 1) -- (10, 1);
\foreach \x in {1, 2, 3, 4, 5, 6.5, 7.3, 8.8} 
    \fill[blue] (\x, 1) circle (2pt);

\node[below] at (1, 0.9) {\textcolor{red}{$x_k$}};
\node[below] at (2, 0.9) {\textcolor{cyan}{$z_\ell$}};
\node[below] at (3, 0.9) {\textcolor{orange}{$z_{j-1}$}};
\node[below] at (4, 0.9) {$x_j$};
\node[below] at (5, 0.9) {\textcolor{orange}{$x_{j-1}$}};
\node[below] at (6.5, 0.9) {$z_j$};
\node[below] at (7.3, 0.9) {\textcolor{cyan}{$x_\ell$}};
\node[below] at (8.8, 0.9) {\textcolor{red}{$z_k$}};

\draw[->, thick, blue] (5, 1.35) -- (5, 1.1);

\node at (-1, 0) {Case 3:};
\draw[thick, ->] (0, 0) -- (10, 0);
\foreach \x in {1, 2, 4, 5, 6.5, 7.3, 8.8} 
    \fill[blue] (\x, 0) circle (2pt);

\node[below] at (1, -0.1) {\textcolor{red}{$x_k$}};
\node[below] at (2, -0.1) {\textcolor{cyan}{$z_\ell$}};
\node[below] at (4, -0.1) {$x_j$};
\node[below] at (5, -0.1) {$z_j$};
\node[below] at (6.5, -0.1) {\textcolor{orange}{$x_{j-1}$}};
\node[below] at (7.3, -0.1) {\textcolor{cyan}{$x_\ell$}};
\node[below] at (8.8, -0.1) {\textcolor{red}{$z_k$}};

\draw[->, thick, blue] (5, 0.35) -- (5, 0.1);

\end{tikzpicture}
\caption{Illustrating the relative rankings of different variables in three cases}
\label{fig:var}
\end{figure}
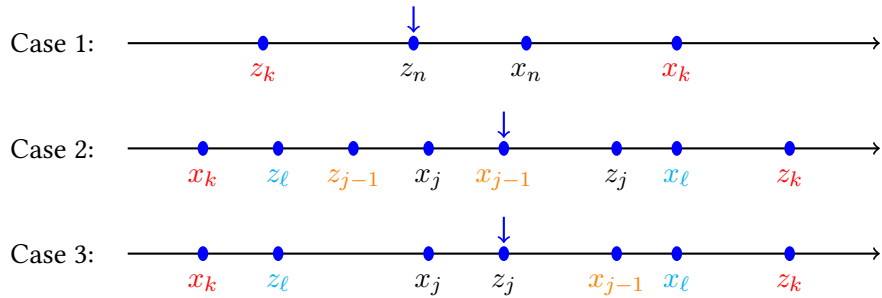

\noindent \textit{Case 1:} $z_{n} \leq x_{n}$. In this case, we set $w_q =  \dots = w_{n} = z_{n}$. Since $z_k \leq z_{n} \leq x_{n} \leq x_k$ for all $k \in [q, n-1]$, there must be $\frac{\partial \ub}{\partial x_k}(x) \leq 0$ for $k \in [q, n]$ and $x \in [z_{n}, x_k]$. Therefore, we have 
\[
R(\vec w) - R(\vec x) ~=~ \sum_{k=q}^{n-1} \int_{x_k}^{z_{n-1}} \frac{\partial \ub}{\partial x_k}(x) \d x ~=~-\sum_{k=q}^{n-1} \int_{z_{n-1}}^{x_k} \frac{\partial \ub}{\partial x_k}(x) \d x ~\geq~ 0.
\]
The first equality is feasible because the function $\frac{\partial \ub}{\partial x_k}$ does not depend on $k' \in [q, n] \setminus \{k\}$, i.e., changing the value of $x_{k'}$ does not change the function $\frac{\partial \ub}{\partial x_k}$.

\noindent \textit{Case 2:} Let $j = \min \{k \in [q, n]: z_k > x_k\}$; $j$ satisfies the condition $x_{j-1} < z_j$. We first note that if we are not in Case 1, at least we have $n \in \{k \in [5, n]: z_k > x_k\}$, and thus $j$ is well defined. In this case, we set $w_q = w_{q+1} =\dots = w_{n} = x_{j-1}$.
For $k \in [j, n]$, we have
\[
z_k ~\geq~ z_j ~>~ x_{j-1} ~\geq~ x_j ~\geq~ x_k.
\]
Therefore, $\frac{\partial \ub}{\partial x_k}(x) \geq 0$ for $k \in [j, n]$ and $x \in [x_k, x_{j-1}]$. On the other hand, for $\ell \in [q, j-1]$, we have
\[
z_\ell ~\leq~ z_{j-1} ~\leq~ x_{j-1} ~\leq~ x_\ell,
\]
where the inequality $z_{j-1} \leq x_{j-1}$ is from the assumption that $j$ is the minimum index that satisfies $z_j > x_j$. Therefore, we have $\frac{\partial \ub}{\partial x_\ell}(x) \leq 0$ for $\ell \in [q, j-1]$ and $x \in [x_{j-1}, x_\ell]$. Then,
\[
R(\vec w) - R(\vec x) ~=~  \sum_{k=j}^{n} \int_{x_k}^{x_{j-1}} \frac{\partial \ub}{\partial x_k}(x) \d x -  \sum_{l=q}^{j - 1} \int_{x_{j-1}}^{x_\ell} \frac{\partial \ub}{\partial x_\ell}(x) \d x~\geq~ 0.
\]

\noindent \textit{Case 3:} Let $j = \min \{k \in [q, n]: z_k > x_k\}$; $j$ satisfies that $z_j \leq x_{j-1}$. In this case, we set $w_q = w_{q+1} =\dots = w_{n-1} = z_j$.
For $k \in [j, n ]$, we have
\[
z_k ~\geq~ z_j ~>~  x_j ~\geq~ x_k.
\]
Therefore, $\frac{\partial \ub}{\partial x_k}(x) \geq 0$ for $k \in [j, n]$ and $x \in [x_k, z_j]$. On the other hand, for $\ell \in [q, j-1]$, we have
\[
z_\ell ~\leq~ z_j ~\leq~ x_{j-1} ~\leq~ x_\ell.
\]
Therefore, we have $\frac{\partial \ub}{\partial x_\ell}(x) \leq 0$ for $\ell \in [q, j-1]$ and $x \in [z_j, x_\ell]$.
Then,
\[
R(\vec w) - R(\vec x) ~=~  \sum_{k=j}^{n} \int_{x_k}^{z_j} \frac{\partial \ub}{\partial x_k}(x) \d x -  \sum_{l=q}^{j - 1} \int_{z_j}^{x_\ell} \frac{\partial \ub}{\partial x_\ell}(x) \d x~\geq~ 0.
\]
In summary, for any vector $\vec x$ with fixed $n, x_1, \dots, x_{q-1}$, we can always update the values $x_q, \dots, x_{n}$ to make sure that $x_q = \dots = x_{n}$ without decreasing the value of $R(\vec x)$. This implies that $R(\vec x)$ is maximized when $x_q  = \dots = x_{n}$ and hence finishes the proof.
\end{proof}

With \Cref{lem:maximized_at_eq_new}, it's sufficient to consider $R_q(\vec x)$ under the assumption that $x_q = x_{q+1} = \cdots = x_n$. Inspired by this, we define the following.
\begin{definition}
We define $\rstar_q(\vec x) (q = 2, 3, \dots)$ to present a further relaxation of $\reg$:
\begin{align*}
\rstar_q(\vec x) ~:=~&  x_1 \cdot \theta_1 + \left(\sum_{i = 2}^{q - 1} (x_1 - x_i) \cdot \frac{1}{i} \cdot (1-\theta_i)^i\right) + \left(\sum_{i = q}^{\infty} (x_1 - x_q) \cdot \frac{1}{i} \cdot (1-\theta_q)^i\right) \\ 
& - \left(\sum_{i = 2}^{q - 1} (x_1 - x_i) \cdot \sum_{k = i+1}^{q-1} \left(\frac{1}{k-1}  - \frac{1}{k}\right) \cdot (1-\theta_k)^k\right)\\
& - \left(\sum_{i = 2}^{q - 1} (x_1 - x_i) \cdot \sum_{k = q}^{\infty} \left(\frac{1}{k-1}  - \frac{1}{k}\right) \cdot (1-\theta_q)^k\right),
\end{align*}
{where the definition comes from taking $x_q = \dots = x_n$ in the function $R_q(\vec x)$, and we further take $n \to \infty$, as one can observe that after taking $x_q = \dots = x_n$, function $R_q$ is monotonely increasing when $n \geq q$ increases. Therefore, the function is maximized when $n \to \infty$.}
\end{definition}

Note that \Cref{lem:maximized_at_eq_new} guarantees that the maximum of $\rstar_q(\vec x)$ is at least the maximum of $\reg(\vec x)$. Therefore, we reduce the problem of finding the hardest instance for \Cref{alg:bopc} from a general $n$-parameters problem to a $q$-parameters problem. As we assumed that $n \geq 10$, it would always be feasible to consider $q = 2, 3, \cdots, 10$. We will further show in the following subsection that taking a constant $q$ is already sufficient to achieve a good regret upper bound.

\subsection{Best-Only Pricing Curve Beats Pricing Curve}
\label{sec:trivial_upper}

Now, we show the power of the relaxation we discussed in \Cref{sec:reg-relax}. We first show that even when $q = 2$, it already beats the $0.25$ lower bound of \Cref{alg:pc} via pricing curve.

\begin{theorem}
\label{thm:trivial_upper}
\Cref{alg:bopc} with pricing curve $f(t) = \exp(-t/c), c = 0.611$ has regret at most $0.23$.
\end{theorem}

\begin{proof}
We prove \Cref{thm:trivial_upper} by showing that $\rstar_2(\vec x)$ for $f(t)$ defined in \Cref{thm:trivial_upper} is upper-bounded by $0.23$. By definition,
\begin{align*}
    \rstar_2(\vec x) ~&=~ x_1 \cdot \theta_1 + (x_1 - x_2) \cdot \sum_{i = 2}^\infty \frac{1}{i} \cdot (1 - \theta_2)^i \\
    ~&=~  x_1 \cdot \theta_1 + (x_1 - x_2) \left(-\ln{\theta_2} - (1 - \theta_2)\right),
\end{align*}
where the last equality uses the Taylor series of $\ln(1 - z) = -\sum_{k = 1}^\infty \frac{z^k}{k}$ for $z \in [0, 1)$.

Take $g(x)$ to be the inverse function of $f(t)$. Then, we have $g(x) = -c \cdot \ln x$. Plugging $\theta_1 = g(x_1)$ into the above equation and taking the partial derivative for $x_1$, we have
\begin{align*}
\frac{\partial \rstar_2(\vec{x})}{\partial x_1}~=~ & -c \ln{x_1} - c  -\ln{\theta_2} - (1 - \theta_2).
\end{align*}

Next, we discuss the following two cases:

\noindent \textit{Case 1: $- c  -\ln{\theta_2} - (1 - \theta_2) > 0$.} In this case, $\frac{\partial \rstar_2(\vec{x})}{\partial x_1}$ is non-negative. Therefore, it's sufficient to only consider $x_1 = 1$ (and therefore $\theta_1 = 1$, and we aim to solve
\begin{align*}
    \max_{x_2}  (1 - x_2) \left(-\ln{\theta_2} - (1 - \theta_2)\right), \quad \text{such that} \quad - c  -\ln{\theta_2} - (1 - \theta_2) > 0.
\end{align*}
Since the above optimization problem only involves single parameter, it's easy to numerically check that the problem is maximized when $x_2 \approx 0.7448$, and the corresponding objective is bounded by $0.229$.

\noindent \textit{Case 2: $- c  -\ln{\theta_2} - (1 - \theta_2) \leq 0$.} In this case, function $\rstar_2(\vec x)$ is maximized when $\frac{\partial \rstar_2(\vec{x})}{\partial x_1} = 0$, which gives 
\[
x_1 = \exp \left(\frac{- c  -\ln{\theta_2} - (1 - \theta_2)}{c}\right).
\]
Then, the problem becomes a single-parameter optimization problem, which aims at maximizing $\rstar_2(\vec x)$ with $x_1$ defining above under the constraint that $- c  -\ln{\theta_2} - (1 - \theta_2) \leq 0$. It's easy to numerically check that the problem is maximized when $x_2 \approx 0.3303$, the corresponding value of $x_1$ is approximately $0.4132$,  and the objective is also bounded by $0.229$.
\end{proof}

\Cref{thm:trivial_upper} suggests that taking $q=2$ is already sufficient to provide a relatively simple analysis for \Cref{alg:bopc} that gets a regret upper bound strictly better than $0.25$. Furthermore, by taking a slightly larger $q$ together with a computer-aided search to find the maximum value of $\rstar_q(\vec x)$, we are able to give an improved upper bound for the regret of \Cref{alg:bopc}. Formally, we give the following theorem:

\label{sec:upper_bound}
\begin{restatable}{theorem}{mainupper}
\label{thm:boupper}
The best-only pricing curve with $f(x) = \exp(-x/c), c = 0.472$ has regret $\leq 0.190$.
\end{restatable}

The proof of \Cref{thm:boupper} takes $q = 5$.
We defer the proof of \Cref{thm:boupper} to \Cref{app:main}.

Finally, we also provide a lower bound for the class of best-only pricing curves:

\begin{restatable}{theorem}{bopclower}
\label{thm:bolower}
No best-only pricing curve (\Cref{alg:bopc}) can achieve a regret less than $0.171$.
\end{restatable}

We defer the proof of \Cref{thm:bolower} to \Cref{app:main}.

\section{Regret Lower Bound for General Algorithms}
\label{sec:lower}

In this section, we provide a general regret lower bound of the secretary problem for all algorithms. To be specific, we prove the following:

\begin{restatable}{theorem}{lowermain}
\label{thm:lower_main}
For secretary problem with regret minimization as the objective, no algorithm can achieve a regret less than $0.152$.
\end{restatable}

Our main idea of proving \Cref{thm:lower_main} is as follows: We first apply Yao's minimax principle, which states that the worst-case regret of any (possibly random) algorithm is at least the optimal regret of a deterministic algorithm on a random input. Next, we show the following \Cref{ex:lower} gives the hard random input we desire:

\begin{restatable}{example}{exlower}
\label{ex:lower}
Let $V = \{a, b, c, 0\}$, where $a = 1$, $b = 0.59$, and $c = 0.38$. The random \cref{ex:lower} is constructed in the following way, with $M := 10^5$ being a sufficiently large integer.
\begin{itemize}
    \item With probability $p_1 = 0.46$, the input is $\I_1$, which contains $V_1 = \{a, b, c\}$ and $M - 3$ values of $0$.
    \item With probability $p_2 = 0.27$, the input is $\I_2$, which contains $V_2 = \{b, c\}$ and $M - 2$ values of $0$.
    \item With probability $p_3 = 0.27$, the input is $\I_3$, which contains $V_3 = \{c\}$ and $M - 1$ values of $0$.
\end{itemize}
\end{restatable}

As \Cref{ex:lower} is explicitly given, our ideas of proving \Cref{thm:lower_main} is to show that the optimal algorithm for \Cref{ex:lower} has the form of \Cref{alg:oai}, and further showing that the regret of \Cref{alg:oai} is at least $0.152$.

\begin{algorithm}[tbh]
\caption{Optimal Algorithm for $\I$}
\label{alg:oai}
Let $A^*(S, v, i) \colon 2^V \times V \times [M] \to \{0, 1\}$ be the optimal decision function \\
Initialize $S \gets \varnothing$ \\
\For{$i = 1 \to M - 1$}
{
    Observe the $i$-th arrival $y_i$ \\
    \lIf{$y_i = 0$}{\textbf{continue}}
    \uIf{$A^*(S, y_i, i) = 1$}
    {Pick $y_i$ and \textbf{exit}}
    \Else{Skip $y_i$ and update $S \gets S \cup \{y_i\}$}
}
Pick $y_M$ if the algorithm has not picked any value
\end{algorithm}

\begin{proof}[Proof of \cref{thm:lower_main}]

The main idea of our proof is to show that the optimal decision function $A^*(S, v, i)$ can be calculated via a backward dynamic programming. Then, the optimality of $A^*(S, v, i)$ guarantees \Cref{alg:oai} is the optimal algorithm for \Cref{ex:lower}.

    We first introduce the following two notations:
    \begin{itemize}
        \item $E(k, S, i)$: we use $E(k, S, i)$ to represent the expected outcome if we run \Cref{alg:oai} starting from the $i$-th arrival, assuming the input of the instance is $\I_k$, subset $S \subseteq V_k$ has arrived before the $i$-th arrival, and the remaining $V_k \setminus S$ will arrive at the arrival slot $i, i+1, \ldots, M$ uniformly at random.  We note that $E(k, S, i)$ is not well defined when $S$ is not a subset of $V_k$, or $|V_k \setminus S| > M-i+1$. For these cases, we set the value of $E(k, S, i)$ be $-\infty$ to mark that this status is invalid.\\
        We will show with induction that the value of $E(k, S, i)$ can be computed explicitly.
        \item $P_k^{S,v,i}$: we define 
        \begin{align}
        \label{eq:cond}
        P_k^{S,v,i} := \left\{ 
        \begin{aligned}
            &  p_k \cdot \frac{P(i-1, |S|) \cdot P(M-i, |V_k| - |S| - 1)}{P(M, V_k)} & \qquad & S \cup \{v\} \subseteq V_k\\
            & 0 &\qquad  & \text{otherwise}
        \end{aligned}
        \right.
        \end{align}
        to be the probability that the following events happen simultaneously: the input for \Cref{alg:oai} is $\I_k$; subset $S \subseteq V_k$ has arrived before the $i$-th arrival; value $v \in V_k$ is the $i$-th arrival.
    \end{itemize}
    
    Now we use backward induction to show \cref{alg:oai} is the optimal algorithm for \cref{ex:lower}. Simultaneously, we show the optimal $A^*(S, v, i)$ is defined by a backward dynamic programming process.
    
    The base case is $i = M$. In this case, the optimal strategy is to take any arriving value, i.e., setting $A^*(S, v, M) = 1$ for every $S$ and $v$ is sufficient. Besides, we have $E(k, S, M) = v$ when $|V_k \setminus S| = 1$, where $v$ is the only value in $V_k \setminus S$, and $E(k, S, M) = 0$ when $S = V_k$.

    For the induction step, we consider the optimal strategy at the $i$-th arrival, assuming the optimal algorithm for \cref{ex:lower} follows \Cref{alg:oai} for arrival slots $i+1, i+2, \ldots, M$. We also assume value $E(k, S, j)$ is computable for $j = i+1, i+2, \ldots, M$. 

    We first decide $A^*(S, 0, i)$. Note that there is no incentive to pick value $0$ when $i \neq M$. Therefore, we set $A^*(S, 0, i) = 0$ for all $S$. Now consider the case that a value $v \neq 0$ arrives at slot $i$. There are two possible options. If we take $v$,  the algorithm gets value $v$. On the other hand, if we skip $v$, the fact that subset $S \cup \{v\}$ has arrived gives a conditional distribution for the future arrivals, and therefore the expected outcome of \cref{alg:oai} in the future is computable. Specifically, if we skip $v$ at slot $i$, the expected outcome of \cref{alg:oai} is exactly
    \begin{align}
    \label{eq:future}
        \sum_{k = 1}^3 \frac{P_k^{S,v,i}}{\sum_{j \in[3]} P_j^{S,v,i}} \cdot E(k, S \cup \{v\}, i).
    \end{align}
    With the help of \cref{eq:future}, the optimal strategy for the $i$-th arrival is clear: if a value $v \in \{a, b, c\}$ becomes the $i$-th arrival, we take it and set $A^*(S, v, i) = 1$ when $v$ is greater than \cref{eq:future}; otherwise, the optimal strategy is to skip $v$, and we set $A^*(S, v, i) = 0$ in this case. 

    Finally, we show the value of $E(k, S, i)$ is computable to finish the induction step. The value of $E(k, S, i)$ is explicitly given by the following expression:
    \begin{align*}
        E(k, S, i) = \sum_{v \in V_k \setminus S} \frac{1}{M - i + 1} \cdot \left(\one[A^*(S, v, i) = 1] \cdot v + \one[A^*(S, v, i) = 0] \cdot E(k, S \cup \{v\}, i + 1) \right).
    \end{align*}
    The backward induction shows that we can determine the optimal algorithm for \cref{ex:lower} by calculating the value of $A^*(S, v, i)$ and $E(k, S, i)$ step-by-step via a backward dynamic programming from $i = M$ to $i = 1$. Therefore, \cref{alg:oai} is the optimal algorithm for \cref{ex:lower}. Following the calculation procedure described above, we find the optimal decision function $A^*(S, v, i)$ as follows:
    \begin{itemize}
        \item $A^*(S, a, i) = 1$ for all $S \not \owns a$, $i \in [M]$.
        \item $A^*(S, v, i) = 0$ for all $S \ni a$, $v \in \{b, c\}$, $i \in [M]$.
        \item $A^*(\varnothing, b, i) = 1$ for $i \in [29396, M]$, while $A^*(\varnothing, b, i) = 0$ for $i \in [1, 29395]$.
        \item $A^*(\varnothing, c, i) = 1$ for $i \in [40051, M]$, while $A^*(\varnothing, c, i) = 0$ for $i \in [1, 40050]$.
        \item $A^*(\{b\}, c, i) = 1$ for $i \in [64026, M]$, while $A^*(\{b\}, c, i) = 0$ for $i \in [1, 64025]$.
        \item $A^*(\{c\}, b, i) = 1$ for $i \in [15538, M]$, while $A^*(\{c\}, b, i) = 0$ for $i \in [1, 15537]$.
    \end{itemize}
    By using computer-aided calculation, it can be further computed that the regret for \cref{ex:lower} with this optimal function $A^*(S, v, i)$ has regret at least $0.1529$, which finishes the proof of \Cref{thm:lower_main}.
\end{proof}

\begin{small}
\bibliographystyle{alpha}
\bibliography{ref.bib,bib.bib}
\end{small}

\appendix

\section{Further Results with Regret Minimization Objective}
\label{sec:revenue}
So far, we have discussed regret minimization for the secretary problem. In this section, we further discuss other combinatorial models with regret minimization as the objective.

\subsection{Multiple-Choice Secretary Model}

We first extend the discussed single-choice secretary model to the multiple-choice secretary variant, in which the algorithm is allowed to choose $k$ values among $n$ sequentially arriving numbers, and the regret minimization objective refers to the difference between the sum of top-$k$ values and the sum of the chosen values, under the assumption that each arriving value falls between $[0, 1]$. For the multiple-choice secretary model, when the objective is to maximize the competitive ratio, \cite{Kleinberg-SODA05} gives the following result.

\begin{proposition}[\cite{Kleinberg-SODA05}]
    For the multiple-choice secretary problem with the objective of optimizing the competitive ratio, there exists a $(1 - O(k^{-1/2}))$-competitive algorithm.
\end{proposition}

Note that when each arriving value falls between $[0, 1]$, the sum of top-$k$ values is at most $k$. Therefore, as a direct corollary, the $(1 - O(k^{-1/2}))$-competitive algorithm immediately gives an $O(\sqrt{k})$ regret bound.

\begin{theorem}
    For the multiple-choice secretary problem with regret minimization as the objective, there exists an algorithm with regret of $O(\sqrt{k})$.
\end{theorem}

\Cref{thm:multiple-lower} below provides a matching lower bound for the multiple-choice secretary problem.

\begin{restatable}{theorem}{kcopieslower}
\label{thm:multiple-lower}
    For the multiple-choice secretary problem with regret minimization as the objective, every algorithm has $\Omega(\sqrt{k})$ worst-case regret.
\end{restatable}

{
\Cref{thm:multiple-lower} follows from a $(1 - \Omega(k^{-1/2}))$-competitive lower bound claimed in \cite{Kleinberg-SODA05}. However, the proof of such lower bound is not explicitly stated in \cite{Kleinberg-SODA05}. Below we provide a direct proof of \Cref{thm:multiple-lower} for completeness.
}

To prove \cref{thm:multiple-lower}, consider the following random instance with $2k$ arriving values. (Without loss of generality, assume $\sqrt{k}$ is an integer and a multiple of $20$.)
\begin{itemize}
\item Case 1: the input contains $k$ ones, $k - \sqrt{k}/5$ zeros, and $\sqrt{k}/5$ one halves.
\item Case 2: the input contains $k - \sqrt{k}/5$ ones, $k$ zeros, and $\sqrt{k}/5$ one halves.
\item The random instance takes Case 1 as input with probability $0.5$, and takes Case 2 as input with probability $0.5$.
\end{itemize}

It is intuitive that no algorithm can distinguish those two cases very well before incurring a regret of $\Omega(\sqrt{k})$. We will formalize this intuition in the proof below.

\begin{proof}[Proof of \cref{thm:multiple-lower}]
We use Yao's minimax principle to prove the lower bound. Specifically, it is sufficient to show that every deterministic algorithm suffers a regret of at least $\sqrt{k} / 1600$ when the input comes from the random instance above.

Consider the following simplification of the problem: the first half of the input (first $k$ values) is revealed to the algorithm at once. The algorithm then determine the strategy of picking values from this first half of the input. Clearly the optimal deterministic algorithm should pick all arriving ones and reject all arriving zeros. It remains to determine the number of one halves to pick. Therefore, the optimal deterministic algorithm can be defined as a function $w(a, b) \colon [0, k] \times [0, \sqrt{k}/5] \to [0, \sqrt{k}/5]$, i.e., when the first half of the input contains $a$ ones and $b$ one halves, the optimal deterministic algorithm decides to take $w(a, b) \leq b$ one halves.

Define $P_1(a, b)$ to be the probability that the first half of the input contains $a$ ones and $b$ one halves when the input comes from Case 1. Similarly, define $P_2(a, b)$ to be the probability that the first half of the input contains $a$ ones and $b$ one halves when the input comes from Case 1. Note that the optimal strategy in hindsight is to pick no one half when the input comes from Case 1, and all $b$ one halves when the input comes from Case 2. Either incorrectly picking a one half in Case 1 or missing a one half in Case 2 suffers a regret of $\frac{1}{2}$. Therefore, the regret of the optimal deterministic algorithm is at least
\begin{align}
    &\sum_{a = 0}^k \sum_{b = 0}^{\sqrt{k}/5} 0.5 \cdot P_1(a, b) \cdot \frac{1}{2} w(a, b) + \sum_{a = 0}^{k - \sqrt{k}/5} \sum_{b = 0}^{\sqrt{k}/5} 0.5 \cdot P_2(a, b) \cdot \frac{1}{2} (b -w(a, b)) \notag \\
    \geq~& \frac{1}{4} \cdot \sum_{a = 0}^k \sum_{b = 0}^{\sqrt{k}/5} \min \{P_1(a, b), P_2(a, b)\} \cdot \max\{w(a, b), b - w(a, b)\} \notag \\
    \geq~&  \sum_{b = 0}^{\sqrt{k}/5}  \frac{b}{8} \cdot \sum_{a = 0}^k  \min \{P_1(a, b), P_2(a, b)\} \notag \\
    \geq~& \sum_{b = \sqrt{k}/20}^{\sqrt{k}/5}  \frac{\sqrt{k}}{160} \cdot \sum_{a = 0}^{\frac{k}{2} - \frac{\sqrt{k}}{10}}  \min \{P_1(a, b), P_2(a, b)\} \label{eq:kcopies1}
\end{align}
where the second inequality uses the fact that $\max\{w(a, b), b-w(a, b)\} \geq \frac{b}{2}$.

Next, we show that $P_1(a, b) \leq P_2(a, b)$, and further simplify the term of $\min \{P_1(a, b), P_2(a, b)\}$ in \cref{eq:kcopies1} to $P_1(a, b)$. Note that
\[
P_1(a, b) = \frac{{k \choose a} \cdot {\sqrt{k}/5 \choose b} \cdot {k - \sqrt{k}/5 \choose k - a - b}}{{2k \choose k}}, \qquad \text{and} \qquad P_2(a, b) = \frac{{k - \sqrt{k}/5 \choose a} \cdot {\sqrt{k}/5 \choose b} \cdot {k \choose k - a - b}}{{2k \choose k}}.
\]
Therefore,
\[
\frac{P_1(a, b)}{P_2(a, b)} ~=~ \frac{{k \choose a}  \cdot {k - \sqrt{k}/5 \choose k - a - b}}{{k - \sqrt{k}/5 \choose a} \cdot  {k \choose k - a - b}} ~=~ \frac{(a + b)! \cdot (k - \sqrt{k}/5 - a)!}{(k - a)! \cdot (a + b - \sqrt{k}/5)!}~=~ \prod_{u = a + b - \sqrt{k} + 1}^{a + b} \frac{u}{u - 2a - b + k}.
\]
When $a \leq \frac{k}{2} - \frac{\sqrt{k}}{10}$, there must be $2a + b \leq k$, i.e., $\frac{u}{u - 2a - b + k} \leq 1$. Therefore, $P_1(a, b) \leq P_2(a, b)$ when $a \leq \frac{k}{2} - \frac{\sqrt{k}}{10}$.

Now we use the above observation to further simplify \cref{eq:kcopies1}. We have
\begin{align}
    \text{\eqref{eq:kcopies1}} ~&=~ \frac{\sqrt{k}}{160} \cdot \sum_{b = \sqrt{k}/20}^{\sqrt{k}/5}   \sum_{a = 0}^{\frac{k}{2} - \frac{\sqrt{k}}{10}}  \min \{P_1(a, b), P_2(a, b)\} \notag\\
    ~&=~ \frac{\sqrt{k}}{160} \cdot \pr\left[b \geq \sqrt{k}/20 \land a \leq k/2 - \sqrt{k}/10 ~\mid \text{input is Case 1}\right] \notag \\
    ~&\geq~ \frac{\sqrt{k}}{160} \cdot \left(1 - \pr\left[b < \sqrt{k}/20 ~\mid \text{input is Case 1} \right] - \pr\left[a > k/2 - \sqrt{k}/10 ~\mid \text{input is Case 1} \right] \right). \label{eq:kcopiestmp}
\end{align}

Now, we assume the input is Case 1, and upper bound the probability that $b < \sqrt{k}/20$ and $a > k/2 - \sqrt{k}/10$ separately. For simplicity, we omit the condition ``input is Case 1'' in the following calculation.

We first show the $\pr\left[b < \sqrt{k}/20\right] \leq \frac{1}{4}$. Observe that for $u < \sqrt{k}/20$, we have
\begin{align}
    \pr[b = u] ~=~ \frac{{\sqrt{k}/5 \choose u} \cdot {2k - \sqrt{k}/5 \choose k - u}}{{2k \choose k}} ~<~ \frac{{\sqrt{k}/5 \choose u + \sqrt{k}/20} \cdot {2k - \sqrt{k}/5 \choose k - u -  \sqrt{k}/20}}{{2k \choose k}} ~=~ \pr[b = u + \sqrt{k}/20], \label{eq:kcopies2}
\end{align}
where the inequality holds because the combination number ${n \choose m}$ is single-peaked with respect to $m$, and is maximized when $m = n/2$. Then, we have ${\sqrt{k}/5 \choose u} < {\sqrt{k}/5 \choose u + \sqrt{k}/20}$ and ${2k - \sqrt{k}/5 \choose k - u} < {2k - \sqrt{k}/5 \choose k - u -  \sqrt{k}/20}$, because  $u < u + \sqrt{k}/20 \leq \sqrt{k}/10$ and $k - \sqrt{k}/10 < k - u - \sqrt{k}/20 < k - u$.

On the other hand, note that 
\begin{align*}
    \pr[b = u] ~=~ \frac{{\sqrt{k}/5 \choose u} \cdot {2k - \sqrt{k}/5 \choose k - u}}{{2k \choose k}} ~=~ \frac{{\sqrt{k}/5 \choose  \sqrt{k}/5 - u} \cdot {2k - \sqrt{k}/5 \choose k - 3\sqrt{k}/20 + u }}{{2k \choose k}} ~=~ \pr[b = \sqrt{k}/5 - u].
\end{align*}
Combining the above equality with \cref{eq:kcopies2}, we have
\begin{align*}
    \pr\left[b < \sqrt{k}/20\right] ~&=~ \sum_{u = 0}^{\sqrt{k}/20 - 1} \pr[b = u] \\
    ~&\leq~ \frac{1}{2} \cdot \sum_{u = 0}^{\sqrt{k}/20 - 1}  \left(\pr[b = u] + \pr[b = u + \sqrt{k}/20] \right) \\
    ~&=~ \frac{1}{2} \cdot \sum_{u = 0}^{\sqrt{k}/10 - 1} \pr[b = u] \\
    ~&=~ \frac{1}{4} \cdot \sum_{u = 0}^{\sqrt{k}/10 - 1}  \left(\pr[b = u] + \pr[b = \sqrt{k}/5 - u] \right) \\
    ~&=~ \frac{1}{4} \cdot \left(\sum_{u = 0}^{\sqrt{k}/10 - 1} \pr[b = u] + \sum_{u = \sqrt{k}/10 + 1}^{\sqrt{k}/5} \pr[b = u] \right) \\
    ~&=~ \frac{1}{4} \cdot \left(1 - \pr[b = \sqrt{k}/10]\right) ~\leq~ \frac{1}{4}.
\end{align*}

Next, we show $\pr\left[a > k/2 - \sqrt{k}/10 \right] \leq 0.65$, or equivalently, we show $\pr\left[a \leq k/2 - \sqrt{k}/10 \right] \geq 0.35$. 

We first upper- and lower- bound the value of ${2m \choose m}$ for any integer $m$ via applying the following Stirling approximation: for all positive integer $m$, we have
\[
\sqrt{2\pi m} \cdot \left(\frac{m}{e}\right)^m \cdot e^{\frac{1}{12m + 1}} ~<~ m! ~<~ \sqrt{2\pi m} \cdot \left(\frac{m}{e}\right)^m \cdot e^{\frac{1}{12m}}.
\]
Plugging the above inequalities to the equality ${2m \choose m} = \frac{(2m)!}{(m!)^2}$, we have
\[
\frac{4^m}{\sqrt{\pi \cdot m}} \cdot \frac{e^\frac{1}{24m}}{e^\frac{2}{12m + 1}} ~<~ {2m \choose m} ~<~ \frac{4^m}{\sqrt{\pi \cdot m}} \cdot \frac{e^\frac{1}{24m + 1}}{e^\frac{1}{6m}}.
\]
We further simplify the above inequality to
\begin{align}
    0.8 \cdot \frac{4^m}{\sqrt{\pi \cdot m}}  ~<~ {2m \choose m} ~<~ \frac{4^m}{\sqrt{\pi \cdot m}} . \label{eq:kcopies3}
\end{align}

Now, we prove $\pr\left[a \leq k/2 - \sqrt{k}/10 \right] \geq 0.35$ with the help of \cref{eq:kcopies3}. Observe that for $u \in [0, k]$, we have
\[
\pr[a = u] ~=~ \frac{{k \choose u}{k \choose k - u}}{{2k \choose k}} ~=~ \pr[a = k - u].
\]
Applying the above equality for all $u \leq k/2 - \sqrt{k}/10$, we have
\[
\pr\left[a \leq k/2 - \sqrt{k}/10 \right] = \sum_{u = 0}^{k/2 - \sqrt{k}/10} \pr[a = u] = \sum_{u = 0}^{k/2 - \sqrt{k}/10} \pr[a = k - u] = \pr\left[a \geq k/2 + \sqrt{k}/10 \right].
\]
Therefore, to show $\pr\left[a \leq k/2 - \sqrt{k}/10 \right] \geq 0.35$, it is sufficient to show 
\[
\pr\left[ k/2 - \sqrt{k}/10 < a < k/2 + \sqrt{k} / 10 \right] \leq 0.3,
\]
which is true because
\begin{align*}
    \pr\left[ k/2 - \sqrt{k} / 10 < a < k/2 + \sqrt{k} / 10 \right] ~&=~ \sum_{u = k/2 - \sqrt{k} / 10 + 1}^{k/2 + \sqrt{k} / 10 - 1} \pr[a = u] \\
    ~&=~ \sum_{u = k/2 - \sqrt{k} / 10 + 1}^{k/2 + \sqrt{k} / 10 - 1} \frac{{k \choose u} \cdot {k \choose k - u}}{{2k \choose k}} \\
    ~&\leq~ \sum_{u = k/2 - \sqrt{k} / 10 + 1}^{k/2 + \sqrt{k} / 10 - 1} \frac{{k \choose k/2} \cdot {k \choose k/2}}{{2k \choose k}} \\
    ~&<~ \frac{\sqrt{k}}{5} \cdot \frac{{k \choose k/2}^2}{{2k \choose k}} \\
    ~&<~ \frac{\sqrt{k}}{5} \cdot \frac{\frac{4^k}{\pi \cdot k/2}}{0.8 \cdot \frac{4^k}{\sqrt{\pi \cdot k}}}~=~ \frac{\sqrt{k}}{5} \cdot \frac{2.5}{\sqrt{\pi \cdot k}} < 0.3,
\end{align*}
where we use \cref{eq:kcopies3} in the last line. Therefore, we have $\pr\left[a \leq k/2 - \sqrt{k}/10 \right] \geq 0.35$, which implies $\pr\left[a > k/2 - \sqrt{k}/10 \right] \leq 0.65$.

Finally, we apply the inequalities $\pr\left[b < \sqrt{k}/20\right] \leq \frac{1}{4}$ and $\pr\left[a > k/2 - \sqrt{k}/10 \right] \leq 0.65$ to \cref{eq:kcopiestmp}, which gives
\[
\text{\cref{eq:kcopiestmp}} ~\geq~ \frac{\sqrt{k}}{160} \cdot (1 - 0.25 - 0.65) ~=~ \frac{\sqrt{k}}{1600},
\]
i.e., every deterministic algorithm suffers a regret of at least $\sqrt{k} / 1600$ from the hard instance, which finishes the proof of \Cref{thm:multiple-lower}.
\end{proof}

\subsection{Revenue Maximization}

The last model we consider is the \emph{revenue} maximization variant of the online stopping problem, in contrast with welfare maximization in earlier sections. Here, a seller is selling one item to $n$ buyers that arrive sequentially. Buyer $i$ has value $x_i \in [0, 1]$. At the time that buyer $i$ arrives, the seller needs to post a price $p_i$. If $x_i \geq p_i$, buyer $i$ buys the item, and the seller gets $p_i$ as revenue, ending the process. As before, the goal is to minimize the \emph{regret}, defined as the difference between (sometimes the expectation of) the best possible revenue $\max_{i \in [n]} x_i$, and the expected revenue collected by the seller.

There are different modelling choices about values and arrivals. The values can be picked by an adversary, or can be drawn stochastically from same/different known distributions. Similarly, the arrival order can be uniformly random or chosen by an adversary. For the positive result of the problem, we show that there exists an algorithm that achieves $1/e$ regret even in the strongest adversary setting.

\begin{restatable}{theorem}{revupper}
    \label{thm:revenue-upper}
    For revenue maximization variant of online stopping problem, there exists an algorithm that achieves $1/e$ regret when buyers have adversarial values.
\end{restatable}

To prove \Cref{thm:revenue-upper}, we present an algorithm -- randomized uniform price (\cref{alg:rup}) -- that can achieve optimal regret of $1/e$ in all of these revenue maximization settings. This algorithm was first discussed in the unpublished work of \cite{KS21} for regret minimization for \emph{welfare} with adversarial arrival order. We further observe that the same algorithm also gives the same regret guarantee for the revenue maximization setting.

\begin{algorithm}[tbh]
\caption{Randomized Uniform Price}
\label{alg:rup}
Let $\alpha \sim U[-1, 0]$, and set uniform price $p \gets e^\alpha$ \\
\For{$\tau = 0 \to 1$ where $x_i$ arrives at $t_i = \tau$}
{
\uIf{$x_i \geq p$}
{
    Pick $x_i$ and \textbf{exit}
}
\Else
{
    Skip $x_i$
}
}
\end{algorithm}

Now, we show that \Cref{alg:rup} is the desired algorithm for \Cref{thm:revenue-upper}.
\begin{proof}[Proof of \Cref{thm:revenue-upper}]
Let $z := \max \{x_1, \ldots, x_n\}$ be the optimal revenue in hindsight. There are two cases:
\begin{itemize}
\item Case (1): $z \leq 1/e$, and the regret is at most $z \leq 1/e$.
\item Case (2): $z > 1/e$. In this case, the item is sold if and only if $z \geq p$, which is equivalent to $\alpha \leq \ln z$. The seller gets revenue of $p = e^\alpha$ if the item is sold. Therefore, the expected revenue of \cref{alg:rup} is
\[
\int_{-1}^{\ln z} e^\alpha \d \alpha = z - \frac{1}{e}.
\]
\end{itemize}
Therefore, the regret is $z - (z - 1/e) = 1/e$.
\end{proof}

For the negative result of the problem, we provide a matching $1/e$ regret lower bound. We also show the hard instance only contains a single stochastic buyer. Therefore the extra stochastic or random order assumption does not make the problem easier.

\begin{restatable}{theorem}{revlower}
    \label{thm:revenue-lower}
    For revenue maximization variant of online stopping problem, there exists an instance in which no algorithm can achieve a regret strictly better than $1/e$ when buyers have adversarial values.
\end{restatable}

\begin{proof}
    Consider a buyer with valuation $x \sim \D$, where the cumulative distribution function of the distribution $\D$ is 
    \begin{align*}
        \label{eq:cond}
        F_{\D}(x) = \left\{ 
        \begin{aligned}
            & 0 &\qquad  & x \in \Big[0, \frac{1}{e}\Big)\\
            & 1 - \frac{1}{ex} &\qquad  & x \in \Big[\frac{1}{e}, 1\Big) \\
            & 1 & \qquad & x = 1
        \end{aligned}
        \right.
    \end{align*}
    Observe that an optimal strategy for the seller is to use any price in $[1/e, 1)$ and get expected revenue of
    \[
    p \cdot (1 - F_{\D}(p)) ~=~ \frac{1}{e}.
    \]
    On the other hand,
    \[
    \mathbb{E}[x] ~=~ \int_0^1 (1 - F(x)) \d x ~=~ \frac{2}{e}.
    \]
    By linearity of expectation, the optimal regret for the seller is $2/e-1/e=1/e$.
\end{proof}

\section{Omitted Proofs in Section 4}
\subsection{Proof of \Cref{thm:boupper}}
\label{app:main}
\mainupper*

\begin{proof}
    Recall that 
    \begin{align*}     
        \rstar_5(\vec x) ~:=~&  x_1 \cdot \theta_1 + \left(\sum_{i = 2}^{4} (x_1 - x_i) \cdot \frac{1}{i} \cdot (1-\theta_i)^i\right) + \left(\sum_{i = 5}^{\infty} (x_1 - x_5) \cdot \frac{1}{i} \cdot (1-\theta_5)^i\right) \\ 
        & - \left(\sum_{i = 2}^{4} (x_1 - x_i) \cdot \sum_{k = i+1}^{4} \left(\frac{1}{k-1}  - \frac{1}{k}\right) \cdot (1-\theta_k)^k\right)\\
        & - \left(\sum_{i = 2}^{4} (x_1 - x_i) \cdot \sum_{k = 5}^{\infty} \left(\frac{1}{k-1}  - \frac{1}{k}\right) \cdot (1-\theta_5)^k\right).
    \end{align*}
    By applying Taylor series of 
    \begin{itemize}
        \item $\ln(1 - z) = -\sum_{k = 1}^{\infty} \frac{z^k}{k}$, for $z \in [0, 1)$, and
        \item $(1 - z)\ln(1 - z) = -z + \sum_{k = 2}^{\infty} \frac{z^k}{k(k-1)}$, for $z \in [0, 1)$,
\end{itemize}
for $z = 1 - \theta_4$, we have
\begin{align*}     
        \rstar_5(\vec x) ~=~&  x_1 \cdot \theta_1 + \left(\sum_{i = 2}^{4} (x_1 - x_i) \cdot \frac{1}{i(i-1)} \cdot (1-\theta_i)^i\right) \\
        &+ (x_1 - x_5) \left(-\ln(\theta_5) - (1-\theta_5) - \frac{(1-\theta_5)^2}{2} - \frac{(1-\theta_5)^3}{3} - \frac{(1-\theta_5)^4}{4}\right) \\ 
        & - \left(\sum_{i = 2}^{4} (x_1 - x_i)\right)\left( \theta_5 \ln(\theta_5) + (1-\theta_5) - \frac{(1-\theta_5)^2}{2} - \frac{(1-\theta_5)^3}{6} - \frac{(1-\theta_5)^4}{12} \right).
\end{align*}
It can be shown via computer-aided search that the maximum of above function is at most $0.190$, which proves \Cref{thm:boupper}.
\end{proof}

\noindent \textbf{Finding the maximum of $\rstar_5(\vec x)$ via heuristic grid search.} We briefly discuss the search algorithm we use for finding the maximum of $\rstar_5(\vec x)$. The goal of \Cref{thm:boupper} is to find a near-optimal function $f$ such that the value of $\max_{\vec x} \rstar_5(\vec x)$ is approximately maximized. To resolve this optimization problem, we use a two-layer heuristic search. The outer layer searches over a certain class of parameterized functions (e.g., function class $f(t) = \exp(-t/c)$, function class $f(t) = (1 - x)^c$, and $10$-folded piecewise linear functions). Once the curve is fixed, we use a second heuristic search in the inner layer to find the optimal parameters in $\vec x$ that maximize $\rstar_5(\vec x)$.

After completing the above experiment, we find that $f(t) = \exp(-t/c)$ with $c = 0.472$ is the near-optimal curve among all the classes of curves we have tried, with an approximate regret upper bound of $0.188$, given by the inner-layer heuristic search algorithm. However, this regret cannot be directly used as a feasible regret upper bound. To verify the true regret upper bound of the curve $f(t) = \exp(-t/0.472)$, we calculate its Lipschitz coefficient, which leads to the following \Cref{lem:lip}:

\begin{restatable}{lemma}{lemlip}
\label{lem:lip}
For two vectors $\vec x = (x_1, x_2, x_3, x_4, x_5)$ and $\vec y = (y_1, y_2, y_3, y_4, y_5)$, if we have $x_1 \geq y_1$, $x_2 \leq y_2$, $x_3 \leq y_3$, $x_4 \leq y_4$, $x_5 \leq y_5$, then 
\begin{align*}
\rstar_5(\vec y) - \rstar_5(\vec x) ~\leq~ &0.472(x_1 - y_1) + 0.444(y_2 - x_2)\\& \qquad + 0.195(y_3 - x_3)  + 0.128(y_4 - x_4) + 0.294(y_5 - x_5).
\end{align*}
\end{restatable}

We defer the proof of \Cref{lem:lip} to \Cref{sec:lip}, and first show how we apply \Cref{lem:lip} to our search algorithm. Our search algorithm is a recursive algorithm. We first initialize the search range $[l_i, r_i] = [f(1), 1]$ for $i \in [5]$, indicating that the search range is a 5-dimension hypercube that contains all vectors satisfying $l_i \leq x_i \leq r_i$ for $i \in [5]$. When the recursive algorithm takes a 5-dimensional hypercube as the input, it first check the value of $\rstar_5(r_1, l_2, l_3, l_4, l_5)$. Then, \Cref{lem:lip} gives an upper bound on all the remaining vectors inside the current search range. If the upper bound is smaller than the target bound $0.190$, the recursive algorithm terminates, as it suggests that the function value of all possible vectors in the hypercube is bounded by $0.190$; otherwise, the algorithm divides each $[l_i, r_i]$ into $[l_i, (l_i + r_i)/2]$ and $[(l_i + r_i)/2, r_i]$, and all possible combinations form $32$ small hypercubes. Then, the algorithm searches each of the 32 small hypercubes recursively.

Note that the search algorithm must terminate if the maximum of $\rstar_5(\vec x)$ is smaller than $0.190$. As the search algorithm we run terminates, it implies the maximum of $\rstar_5(\vec x)$ is smaller than $0.190$.

\subsection{proof of \Cref{lem:lip}}
\label{sec:lip}

For simplicity of the analysis, we use the original version of 
\begin{align*}     
        \rstar_5(\vec x) ~=~&  x_1 \cdot \theta_1 + \left(\sum_{i = 2}^{4} (x_1 - x_i) \cdot \frac{1}{i} \cdot (1-\theta_i)^i\right) + \left(\sum_{i = 5}^{\infty} (x_1 - x_5) \cdot \frac{1}{i} \cdot (1-\theta_5)^i\right) \\ 
        & - \left(\sum_{i = 2}^{4} (x_1 - x_i) \cdot \sum_{k = i+1}^{4} \left(\frac{1}{k-1}  - \frac{1}{k}\right) \cdot (1-\theta_k)^k\right)\\
        & - \left(\sum_{i = 2}^{4} (x_1 - x_i) \cdot \sum_{k = 5}^{\infty} \left(\frac{1}{k-1}  - \frac{1}{k}\right) \cdot (1-\theta_5)^k\right).
    \end{align*}
    to prove \Cref{lem:lip}. Note that to prove \Cref{lem:lip}, it's sufficient to prove that $\frac{\partial \rstar_5(\vec x)}{\partial x_1} \geq -0.472$, $\frac{\partial \rstar_5(\vec x)}{\partial x_i}\leq 0.444,0.195,0.128,0.294$ for $i = 2, 3, 4, 5$, respectively. Then, a multiple integral is sufficient to prove \Cref{lem:lip}.

\noindent \textbf{Bounding $\frac{\partial \rstar_5(\vec x)}{\partial x_1}$.} We have \begin{align*}
        \frac{\partial \rstar_5(\vec x)}{\partial x_1} ~=&~ -c +\theta_1 +  \sum_{i = 2}^{4}  \frac{1}{i} \cdot (1-\theta_i)^i + \sum_{i = 5}^{\infty} \frac{1}{i} \cdot (1-\theta_5)^i \\ 
        & - \sum_{k = 2}^{4} (k-2)\left(\frac{1}{k-1}  - \frac{1}{k}\right) \cdot (1-\theta_k)^k\\
        & - 3\sum_{k = 5}^{\infty} \left(\frac{1}{k-1}  - \frac{1}{k}\right) \cdot (1-\theta_5)^k\\
        =&~ -c +\theta_1 +  \sum_{i = 2}^{4}  \frac{1}{i(i-1)} \cdot (1-\theta_i)^i + \sum_{i = 5}^{\infty} \frac{i-4}{i(i-1)} \cdot (1-\theta_5)^i \\ 
        \geq& -c ~=~ -0.472.
    \end{align*}

\noindent \textbf{Bounding $\frac{\partial \rstar_5(\vec x)}{\partial x_2}$.} We have
\begin{align*}
        \frac{\partial \rstar_5(\vec x)}{\partial x_2} ~=&~  -\frac{1}{2} \cdot (1-\theta_2)^2 - (x_1-x_2)(1-\theta_2)\theta'_2 \\
        &+\sum_{k = 3}^{4} \left(\frac{1}{k-1}  - \frac{1}{k}\right) \cdot (1-\theta_k)^k+\sum_{k = 5}^{\infty} \left(\frac{1}{k-1}  - \frac{1}{k}\right) \cdot (1-\theta_5)^k \\
        \leq & -\frac{1}{2} \cdot (1-\theta_2)^2 +(1-x_2)(1-\theta_2)\frac{c}{x_2}\\
        &+ \frac{1}{6} \cdot (1-\theta_2)^3 + \frac{1}{12} \cdot (1-\theta_2)^4 + \frac{1}{4} \cdot (1-\theta_2)^5 ~\leq~ 0.444,
\end{align*}
where the first inequality is from the observation that 
\begin{align}
\label{eq:key-obs}
        \sum_{k = 5}^{\infty} \left(\frac{1}{k-1}  - \frac{1}{k}\right) \cdot (1-\theta_5)^k \leq \sum_{k = 5}^{\infty} \left(\frac{1}{k-1}  - \frac{1}{k}\right) \cdot (1-\theta_5)^5 = \frac14 (1-\theta_5)^5,
\end{align}
and the last inequality can be determined via solving a single-parameter optimization problem.
    
\noindent \textbf{Bounding $\frac{\partial \rstar_5(\vec x)}{\partial x_3}$.} We have
    \begin{align*}
        \frac{\partial \rstar_5(\vec x)}{\partial x_3} ~=&~ - \frac{1}{3} \cdot (1-\theta_3)^3  - (x_1-x_3)(1-\theta_3)^2\theta'_3 + \frac{1}{12} \cdot (1-\theta_4)^4\\
        &+ \left((x_1 - x_2) \cdot\frac12 \cdot (1-\theta_3)^2 \cdot \theta'_3 \right) + \sum_{k = 5}^{\infty} \left(\frac{1}{k-1}  - \frac{1}{k}\right) \cdot (1-\theta_5)^k\\
        ~\leq&~ - \frac{1}{3} \cdot (1-\theta_3)^3 + (1-x_3)(1-\theta_3)^2\frac{c}{x_3} + \frac{1}{12} \cdot (1-\theta_3)^4 + \frac14 \cdot (1-\theta_3)^5\\
        ~\leq&~ 0.195,
    \end{align*}
    where the first inequality follows from \Cref{eq:key-obs} together with the fact that $(x_1 - x_2) \cdot\frac12 \cdot (1-\theta_3)^2 \cdot \theta'_3 \leq 0$, and the last inequality can be given via solving a single-parameter optimization problem.

\noindent \textbf{Bounding $\frac{\partial \rstar_5(\vec x)}{\partial x_4}$.} We have
    \begin{align*}
        \frac{\partial \rstar_5(\vec x)}{\partial x_4} ~=&~ -\frac{1}{4} \cdot (1-\theta_4)^4 - (x_1-x_4)(1-\theta_4)^3\theta'_4\\
        &+ \left( (2x_1 - x_2 - x_3) \cdot \frac1{3} \cdot (1-\theta_4)^3 \theta_4'\right)+\sum_{k = 5}^{\infty} \left(\frac{1}{k-1}  - \frac{1}{k}\right) \cdot (1-\theta_5)^k\\
        ~\leq&~ -\frac{1}{4} \cdot (1-\theta_4)^4 + (1-x_4)(1-\theta_4)^3\frac{c}{x_4}
        +\frac14 (1-\theta_4)^5\\
        ~\leq&~0.128,
    \end{align*}
    where the first inequality follows from \Cref{eq:key-obs} together with the fact that $ (2x_1 - x_2 - x_3) \cdot \frac1{3} \cdot (1-\theta_4)^3 \theta_4' \leq 0$, and  the last inequality can be given via solving a single-parameter optimization problem.

\noindent \textbf{Bounding $\frac{\partial \rstar_5(\vec x)}{\partial x_4}$.} We have
    \begin{align*}
        \frac{\partial \rstar_5(\vec x)}{\partial x_5} ~=&~ -\sum_{i = 5}^{\infty}  \frac{1}{i} \cdot (1-\theta_5)^i - \sum_{i = 5}^{\infty} (x_1 - x_5)  \cdot (1-\theta_5)^{i-1} \cdot \theta'_5 \\ 
        & + \left((3x_1 - x_2 - x_3 - x_4) \cdot \sum_{k = 5}^{\infty} \frac{1}{k-1}  \cdot (1-\theta_5)^{k-1}\cdot \theta'_5\right)\\
        ~\leq&~ -\sum_{i = 5}^{\infty} \frac{1}{i} \cdot (1-\theta_5)^i + \sum_{i = 5}^{\infty} (1 - x_5) \cdot (1-\theta_5)^{i-1} \cdot \frac{c}{x_5} \\
        =&\left(\ln(\theta_5) + (1-\theta_5) + \frac{(1-\theta_5)^2}{2} + \frac{(1-\theta_5)^3}{3} - \frac{(1-\theta_5)^4}{4}\right) \\
        &+ (1 - x_5) \cdot (1-\theta_5)^4 \cdot \frac{c}{x_5} \frac1{\theta_5} \\
        ~\leq &~ 0.294,
    \end{align*}
    where the first inequality follows from \Cref{eq:key-obs} together with the fact that $(3x_1 - x_2 - x_3 - x_4) \cdot \sum_{k = 5}^{\infty} \frac{1}{k-1}  \cdot (1-\theta_5)^{k-1}\cdot \theta'_5 \leq 0$, and  the last inequality can be given via solving a single-parameter optimization problem.

\subsection{Proof of \Cref{thm:bolower}}

\bopclower*

To prove \Cref{thm:bolower}, we show the following \Cref{ex:lower_bopc} is the hard instance for \Cref{alg:bopc}.

\begin{example}
\label{ex:lower_bopc} 
Let $a = 1$, $b = 0.61$, $c = 0.48$, and $d = 0.44$. 
We consider the following group of instances:
\begin{itemize}
    \item For instance $\I_1$, the input is $\vec x_1 = (a, b, c, d)$.
    \item For instance $\I_2$, the input is $\vec x_2 = (b, c, d)$.
    \item For instance $\I_3$, the input is $\vec x_3 = (c, d)$.
    \item For instance $\I_4$, the input is $\vec x_4 = (d)$.
\end{itemize}
\end{example}

\begin{proof}[Proof of \Cref{thm:bolower}]

We show that for every best-only pricing curve, its worst performance among the four instances $\I_1, \I_2, \I_3, \I_4$ introduced in \Cref{ex:lower_bopc} is at least $0.171$.

Similar to our derivation of the upper bound, we again start at the following expression of \cref{lem:bopc_eq}.
\begin{align}
\label{eq:bopc_lower}
    \reg(\vec x) ~=~& x_1 \cdot \theta_1 + \left(\sum_{i = 2}^{n} (x_1 - x_i) \cdot \frac{1}{i} \cdot (1-\theta_i)^i\right) \nonumber \\
    &- \left(\sum_{i = 2}^{n} (x_1 - x_i) \cdot \sum_{k = i+1}^{n} \left(\frac{1}{k-1} - \frac{1}{k}\right) \cdot (1-\theta_k)^k\right).
\end{align}

Now we show that the optimal pricing curve for \cref{ex:lower_bopc} achieves a regret at least $0.171$. The key observation is that given a pricing curve $f(t)$ with $g(x)$ being the ``inverse function'' of $f(t)$, we only need the values of parameters $\theta_a := g(a), \theta_b := g(b), \theta_c := g(c)$, and $\theta_d := g(d)$ to compute the expected regret $\mathbb{E}_{\vec x \sim \D}[\reg(\vec x)]$. Therefore, instead of searching the optimal pricing curve, it is sufficient to search the optimal parameters $\theta^*_a, \theta^*_b, \theta^*_c, \theta^*_d$ that minimizes the expected regret of \cref{eq:bopc_lower}. The numerical experiment shows that the optimal parameters are $\theta^*_a = 0$, $\theta^*_b \approx 0.21728$, $\theta^*_c \approx 0.33677$, $\theta^*_d \approx 0.38633$. The worst regret with these optimal parameters is at least $0.1712 > 0.171$, which implies that no best-only pricing curve has regret less than $0.171$.
\end{proof}

\end{document}